\documentclass[journal]{IEEEtran}

\usepackage{cite}
\ifCLASSINFOpdf
\usepackage[pdftex]{graphicx}
\else
\usepackage[dvips]{graphicx}
\fi
\usepackage{amsmath,amssymb,amsthm,mathrsfs,amsfonts,dsfont,stmaryrd}
\usepackage{array}
\usepackage{dsfont}
\usepackage{amsbsy}
\usepackage{epstopdf}
\usepackage{graphicx,color}

\usepackage{graphicx}

\usepackage{url}

\newtheorem{lemma}{Lemma}
\newtheorem{proposition}{Proposition}
\newtheorem{remark}{Remark}

\long\def\symbolfootnote[#1]#2{\begingroup%
	\def\thefootnote{\fnsymbol{footnote}}\footnote[#1]{#2}\endgroup}
\newtheorem{theorem}{Theorem}
\newtheorem{definition}{Definition}

\newcommand{\dv}{\mathbf} 
\newcommand{\mc}{\mathcal} 
\newcommand{\mb}{\mathbf} 
\newcommand{\mkv}{-\!\!\!\!\minuso\!\!\!\!-}

\allowdisplaybreaks

\usepackage{dblfloatfix}

\newcommand{\bqed}{\tag*{$\blacksquare$}}
\newcommand*{\qedblack}{\hfill\ensuremath{\blacksquare}}

\newcommand{\squeezeup}{\vspace{-2em}}
\newcommand{\squeezeuptwo}{\vspace{-1.5em}}

\newcommand{\E}{\mathbb{E}}

\begin{document}
\fontencoding{OT1}\fontsize{9.4}{11}\selectfont

\title{Vector Gaussian CEO Problem Under \\ Logarithmic Loss} 
\author{
Yi{\u{g}}it U{\u{g}}ur $^{\dagger}$$^{\ddagger}$ \qquad \quad I\~naki Estella Aguerri $^{\dagger}$ \qquad \quad Abdellatif Zaidi $^{\dagger}$$^{\ddagger}$  \\   
{\small
$^{\dagger}$ Mathematics and Algorithmic Sciences Lab., France Research Center, Huawei Technologies, Boulogne-Billancourt, 92100, France\\
$^{\ddagger}$ Universit\'e Paris-Est, Champs-sur-Marne, 77454, France\\ \vspace{-1mm}
\{yigit.ugur@huawei.com,  inaki.estella@huawei.com, abdellatif.zaidi@u-pem.fr\} \vspace{-3em}} 
} 

\maketitle

\begin{abstract} 
In this paper, we study the vector Gaussian Chief Executive Officer (CEO) problem under logarithmic loss distortion measure. Specifically, $K \geq 2$ agents observe independently corrupted Gaussian noisy versions of a remote vector Gaussian source, and communicate independently with a decoder or CEO over rate-constrained noise-free links. The CEO wants to reconstruct the remote source to within some prescribed distortion level where the incurred distortion is measured under the logarithmic loss penalty criterion. We find an explicit characterization of the rate-distortion region of this model. For the proof of this result, we obtain an outer bound on the region of the vector Gaussian CEO problem by means of a technique that relies on the de Bruijn identity and the properties of Fisher information. The approach is similar to Ekrem-Ulukus outer bounding technique for the vector Gaussian CEO problem under quadratic distortion measure, for which it was there found generally non-tight; but it is shown here to yield a complete characterization of the region for the case of logarithmic loss measure. Also, we show that Gaussian test channels with time-sharing exhaust the Berger-Tung inner bound, which is optimal. Furthermore, we also show that the established result under logarithmic loss provides an outer bound for a quadratic vector Gaussian CEO problem with \textit{determinant} constraint, for which we characterize the optimal rate-distortion region.   
\end{abstract}

\vspace{-1.5em}
\section{Introduction}\label{secI}

Consider the vector Gaussian Chief Executive Officer (CEO) problem shown in Figure~\ref{fig-system-model-CEO}. In this model, there is an arbitrary number $K \geq 2$ of agents each having a noisy observation of a vector Gaussian source $\dv X$. The goal of the agents is to describe the source to a central unit, which wants to reconstruct this source to within a prescribed distortion level. The incurred distortion is measured according to some loss measure $d\: :\: \mc X \times \hat{\mc X} \rightarrow \mathbb{R}$, where $\hat{\mc X}$ designates the reconstruction alphabet. For quadratic distortion measure, i.e., 
\begin{equation*}
d(x,\hat{x})=|x-\hat{x}|^2,
\end{equation*}
the rate-distortion region of the vector Gaussian CEO problem is still unknown in general, except in few special cases the most important of which is perhaps the case of scalar sources, i.e., scalar Gaussian CEO problem, for which a complete solution, in terms of characterization of the the optimal rate-distortion region, was found independently by Oohama in~\cite{O05} and by Prabhakaran \textit{et al.} in~\cite{PTR04}. Key to establishing this result is a judicious application of the entropy power inequality. The extension of this argument to the case of vector Gaussian sources, however, is not straightforward as the entropy power inequality is known to be non-tight in this setting. The reader may refer also to~\cite{CW11, WC12} where non-tight outer bounds on the rate-distortion region of the vector Gaussian CEO problem under quadratic distortion measure are obtained by establishing some extremal inequalities that are similar to Liu-Viswanath~\cite{LV07}, and to~\cite{XW16} where a strengthened extremal inequality yields a complete characterization of the region of the vector Gaussian CEO problem in the special case of trace distortion constraint.

\begin{figure}[t!]
\centering
\includegraphics[width=\linewidth]{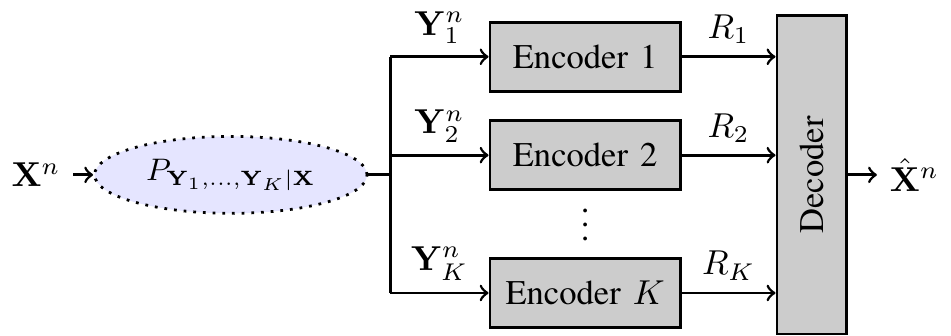}
\squeezeup
\caption{Chief Executive Officer (CEO) source coding problem.}	
\squeezeuptwo
\label{fig-system-model-CEO}
\end{figure}

In this paper, we study the CEO problem of Figure~\ref{fig-system-model-CEO} in the case in which $(\dv X, \dv Y_1,\ldots,\dv Y_K)$ is jointly Gaussian and the distortion is measured using the logarithmic loss criterion, i.e., 
\begin{equation}
d^{(n)}(x^n, \hat{x}^n) = \frac{1}{n} \sum\nolimits_{i=1}^{n} d(x_i,\hat{x}_i),
\label{definition-log-loss-distortion-measure-n-letter}
\end{equation}
with the letter-wise distortion given by  
\begin{equation}
d(x,\hat{x}) = \log\Big(\frac{1}{\hat{x}(x)}\Big),
\label{definition-log-loss-distortion-measure}
\end{equation}
where $\hat{x}(\cdot)$ designates a probability distribution on $\mc X$, and $\hat{x}(x)$ is the value of this distribution evaluated for the outcome $x \in \mc X$. 

The logarithmic loss distortion measure, often referred to as \textit{self-information loss} in the literature about prediction, plays a central role in settings in which reconstructions are allowed to be `soft', rather than `hard' or deterministic. That is, rather than just assigning a deterministic value to each sample of the source, the decoder also gives an assessment of the degree of confidence or reliability on each estimate, in the form of weights or probabilities. This measure, which was introduced in the context of rate-distortion theory by Courtade \textit{et al.}~\cite{CW11-2, CW14}, has appreciable mathematical properties~\cite{JCVW15,NW15}, such as a deep connection to lossless coding for which fundamental limits are well developed (e.g., see~\cite{SRV17} for recent results on universal lossy compression under logarithmic loss that are built on this connection). Also, it is widely used as a penalty criterion in various contexts, including clustering and classification~\cite{TPB99}, pattern recognition, learning and prediction~\cite{C-BL06}, image processing~\cite{AABG06}, secrecy~\cite{KCOSW16} and others.

The main contribution of this paper is a complete characterization of the rate-distortion region of the vector Gaussian CEO problem of Figure~\ref{fig-system-model-CEO} under logarithmic loss distortion measure. The  result can be seen as the counterpart, to the vector Gaussian case, of that by Courtade and Weissman~\cite[Theorem 3]{CW14} who established the rate-distortion region of the CEO problem under logarithmic loss in the discrete memoryless (DM) case. For the proof of this result, we derive an outer bound on the rate-distortion region of the vector Gaussian CEO problem by evaluating the outer bound from the DM model using the de Bruijn identity, a connection between differential entropy and Fisher information, along with the properties of minimum mean square error (MMSE) and Fisher information. By opposition to the case of quadratic distortion measure, for which the application of this technique was shown in~\cite{EU14} to result in an outer bound that is generally non-tight, we show that this approach is successful in the case of logarithmic distortion measure and yields a complete characterization of the region. The proof of the achievability part simply corresponds to the evaluation of the result for the DM model using Gaussian test channels and \textit{no} time-sharing. While this does \textit{not} imply that Gaussian test channels also exhaust the Berger-Tung inner bound for this model, we show that they \textit{do} but might generally require time-sharing. Furthermore, we also show that the established result under logarithmic loss provides an outer bound for a quadratic vector Gaussian CEO problem with \textit{determinant} constraint, for which we characterize the optimal rate-distortion region (see~\cite{PCL03,SSBGS02} for examples of usage of this determinant constraint in the context of equalization).

In the case of one agent, i.e., the remote vector Gaussian Wyner-Ziv model under logarithmic loss, the model was resolved in~\cite{TC09}; and, so, our result here generalizes that of~\cite{TC09} to the case of an arbitrarily number of agents. Related to this aspect, it is also worth mentioning that the orthogonal transform technique which was used in~\cite{TC09} to reduce the vector setting to one of parallel scalar Gaussian settings seems insufficient to diagonalize all the noise covariance matrices simultaneously in the case of more than one agent.     

\textit{Notation:} Throughout, we use the following notation. Upper case letters denote random variables, e.g., $X$; lower case letters denote realizations of random variables, e.g.,  $x$; and calligraphic letters denote sets, e.g., $\mathcal{X}$.  The cardinality of a set $\mc X$ is denoted by $|\mc X|$. A length-$n$ sequence $(X_1,\ldots,X_n)$ is denoted as  $X^n$. Boldface upper case letters denote vectors or matrices, e.g., $\dv X$, where context should make the distinction clear. For an integer $K \geq 1$, we denote the set of integers smaller or equal $K$ as $\mc K$. For a set of integers $\mc S \subseteq \mc K$, the notation $X_{\mc S}$ designates the set of random variables $\{X_k\}$ with indices in the set $\mc S$, i.e., $X_{\mc S}=\{X_k\}_{k \in \mc S}$.

In this paper, due to space limitations some of the proofs are omitted or only outlined. Detailed proofs as well as the extension of the results of this paper to the case in which the decoder also has its own correlated side information stream can be found in~\cite{PaperFull}.

\vspace{-1em}
\section{Problem Formulation}\label{secII}

Consider the $K$-encoder CEO problem shown in Figure~\ref{fig-system-model-CEO}. In this paper, the agents' observations are assumed to be Gaussian noisy versions of a remote vector Gaussian source. Specifically, let $(\dv X, \dv Y_1,\ldots,\dv Y_K)$ be a jointly Gaussian random vector, with zero mean and covariance matrix $\dv \Sigma_{[\dv x, \dv y_1, \ldots, \dv y_K]}$. The vector $ \dv X \in \mathbb{C}^{n_x}$ is complex-valued, and has $n_x \in \mathbb{N}$ dimensions; and vector $\dv Y_k \in \mathbb{C}^{n_k}$, $k=1,\ldots,K$, is complex-valued and has $n_k \in \mathbb{N}$ dimensions. Throughout, it is assumed that the following Markov chain holds
\vspace{-0.15em}
\begin{equation}
\dv  Y_k \mkv  \dv X \mkv \dv Y_{\mathcal{K}/k}, \qquad k = 1,\ldots,K.
\label{eq:MKChain_pmf}
\end{equation}\\[-1.2em]
Let now $\{(\dv X_i, \dv Y_{1,i}, \ldots, \dv Y_{K,i})\}^n_{i=1}$ be a sequence of $n$ independent and identically distributed (i.i.d.) copies of  $(\dv X, \dv Y_1,\ldots,\dv Y_K)$.  Encoder $k$, $k=1,\ldots,K$, observes $\dv Y^n_k := \dv Y^n_{k,1}$. Using~\eqref{eq:MKChain_pmf}, in what follows we assume without loss of generality that  
\vspace{-0.15em}
\begin{equation*}
\mathbf{Y}_{k,i} = \mathbf{H}_{k}\mathbf{X}_i+\mathbf{N}_{k,i},\qquad i = 1,\ldots,n,
\label{mimo-gaussian-model}
\end{equation*}\\[-1.15em]
where $\mathbf{H}_{k} \in \mathbb{C}^{n_k\times n_x}$ designates the channel that connects $\dv X_i$ to $\dv Y_{k,i}$ and $\mathbf{N}_{k,i} \in \mathbb{C}^{n_k}$ is an $n_k$-dimensional, complex-valued, vector Gaussian noise with zero-mean and covariance matrix $\dv\Sigma_k$. All noises $\dv N_{k,i}$ are independent among them, and from $\dv X_i$. 

Encoder $k$, $k=1,\ldots,K$, uses $R_k$ bits per sample to describe its observation $\dv Y^n_k$ to the decoder. The decoder wants to reproduce a soft-estimate of the  remote source $\dv X^n \in \mathbb{C}^{n\times n_x}$. That is, we consider the reproduction alphabet $\hat{\mc X}$ to be equal to the set of probability distributions over the source alphabet $\mathbb{C}^{n\times n_x}$ and the distortion measure is the logarithmic loss criterion as defined by~\eqref{definition-log-loss-distortion-measure-n-letter}.

\begin{definition}
A rate-distortion code (of blocklength $n$) for the CEO problem consists of $K$ encoding functions
\begin{equation*}
\phi^{(n)}_k \: : \:  \mathbb{C}^{n \times n_k} \rightarrow \{1,\ldots,M^{(n)}_k\}, \quad  k=1,\ldots,K,
\end{equation*}
and a decoding function
\begin{equation*}
\psi^{(n)} \: : \:  \{1,\ldots,M^{(n)}_1\} \times \ldots \times \{1,\ldots,M^{(n)}_K\} \rightarrow \hat{\mc X}^n,
\end{equation*}
where $\hat{\mc X}^n$ designates the set of probability distributions over the $n$-Cartesian product of $\mathbb{C}^{n_x}$. \qedblack
\end{definition}

\vspace{-1em}
\begin{definition}
A rate-distortion tuple $(R_1,\ldots,R_K,D)$ is achievable for the vector Gaussian CEO problem if there exist a blocklength  $n$, $K$ encoding functions $\{\phi^{(n)}_k\}_{k=1}^K$ and a decoding function $\psi^{(n)}$ such that
\begin{align*}
R_k &\geq \frac{1}{n}\log M^{(n)}_k, \quad k=1,\ldots,K, \\
D &\geq \mathbb{E}[d(\dv X^n,\psi^{(n)}\big(\phi^{(n)}_1(\dv Y^n_1), \ldots, \phi^{(n)}_K(\dv Y^n_K)\big))].
\end{align*} 
	
\noindent The rate-distortion region $\mc{RD}_{\mathrm{L}}^\star$ of the vector Gaussian CEO problem under logarithmic loss is defined as the union of all non-negative tuples $(R_1,\ldots,R_K,D)$ that are achievable. \qedblack
\end{definition}

\vspace{-0.5em}
One important goal in this paper is to characterize the rate-distortion region $\mc{RD}_{\mathrm{L}}^{\star}$. 

\vspace{-2.5mm}
\section{Vector Gaussian CEO Problem Under Logarithmic Loss}\label{secIII}

\vspace{-0.5mm}
\subsection{Rate-Distortion Region}\label{secIII_subsecA}

The rate-distortion region of the discrete memoryless $K$-encoder CEO problem under logarithmic loss has been fully characterized by Courtade-Weissman in~\cite[Theorem 10]{CW14} in the case in which the Markov chain~\eqref{eq:MKChain_pmf} holds. This result can be extended to the case of Gaussian sources as we stated in the following proposition.

\begin{definition}~\label{def:GaussSumCap_I}
For given tuple of auxiliary random variables $(U_1,\ldots,U_K,Q)$ with distribution $p(u_1,\ldots,u_K,q)$ such that $p(\dv x,\dv y_1,\ldots,\dv y_K,u_1,\ldots,u_K,q)$ factorizes as	
\begin{equation}~\label{joint-measure-auxiliary-random-variables-CEO-problem}
p(q) p(\dv x) \prod\nolimits_{k=1}^K p(\dv y_k|\dv x) \prod\nolimits_{k=1}^K p(u_k|\dv y_k,q),
\end{equation}
$\mc{RD}_\mathrm{L}^\mathrm{I}(U_1,\ldots,U_K,Q)$ denotes the set of all non-negative tuples $(R_1,\ldots,R_K,D)$ that satisfy, for all subsets $\mc S \subseteq \mc K$,
\begin{equation}~\label{eq:GaussSumCap_1} 
\sum\nolimits_{k \in \mathcal{S}} R_k + D \geq \sum\nolimits_{k\in\mathcal{S}} I(\dv Y_{k};U_{k}|\dv X,Q)  + h(\dv X| U_{\mathcal{S}^c},Q). 
\end{equation}
{\fontsize{9.3}{11}\selectfont Also, let $\mc{RD}_\mathrm{L}^\mathrm{I} := \bigcup \mc{RD}_\mathrm{L}^\mathrm{I} (U_1,\ldots,U_K,Q)$ where the union is taken over all tuples $(U_1,\ldots,U_K,Q)$ with distributions that satisfy~\eqref{joint-measure-auxiliary-random-variables-CEO-problem}.} \qedblack
\end{definition}

\begin{proposition}~\label{proposition-continous-RD1-CEO}
$\mc{RD}_\mathrm{L}^\star = \mc{RD}_\mathrm{L}^\mathrm{I}$.  
\end{proposition}

\begin{proof}
The proof of Proposition~\ref{proposition-continous-RD1-CEO} is given in Section~\ref{proof-continous-RD1-CEO}.
\end{proof}

One main result in this paper is an explicit characterization of $\mc {RD}^{\star}_{\mathrm{L}}$. To this end, we show that the region $\mc{RD}_\mathrm{L}^\mathrm{I}$ is exhausted by Gaussian test channels. Also, we show that one can optimally set $Q=\emptyset$, i.e., time-sharing is not needed.

\begin{theorem}~\label{th:GaussSumCap}
The rate-distortion region $\mc {RD}_{L}^{\star}$ of the vector Gaussian CEO problem under logarithmic loss is given by the set of all non-negative rate-distortion tuples $(R_1,\ldots,R_K,D)$ that satisfy, for all subsets $\mc S \subseteq \mc K$,
\vspace{-0.3em}
\begin{align*}
D + \sum\nolimits_{k\in \mathcal{S}}R_k  \geq&  \sum\nolimits_{k\in \mathcal{S}}\log\frac{1}{ |\dv I- \dv\Omega_k \dv\Sigma_k|} \nonumber\\
& + \log \big|(\pi e) \big(\dv\Sigma_{\dv x}^{-1} + \sum\nolimits_{k\in\mathcal{S}^{c}}\mathbf{H}_{k}^\dagger  \dv\Omega_k \mathbf{H}_{k} \big)^{-1} \big| ,
\end{align*}
for some matrices $\{\dv\Omega_k\}_{k=1}^K$ such that $\dv 0 \preceq \dv\Omega_k \preceq \dv\Sigma_k^{-1}$. 
\end{theorem}

\begin{proof} 
The proof of the direct part of Theorem~\ref{th:GaussSumCap} follows simply by evaluating~\eqref{eq:GaussSumCap_1} using Gaussian test channels and no time-sharing. Specifically, we set $Q= \emptyset$ and $p( u_k|\dv y_k,q) = \mc{CN}(\dv y_k,  \mathbf{\Sigma}_{k}^{1/2}(\dv\Omega_k-\dv I)\mathbf{\Sigma}_{k}^{1/2}) $. The proof of the converse  appears in Section~\ref{secV_subsecA}.
\end{proof}

\vspace{-1em}
\begin{remark}
In~\cite{CW14}, it was shown that the union of all rate-distortion tuples that satisfy~\eqref{eq:GaussSumCap_1} for all subsets $\mc S \subseteq \mc K$ coincides with the Berger-Tung inner bound in which time-sharing is used. The direct part of Theorem~\ref{th:GaussSumCap} is obtained by evaluating~\eqref{eq:GaussSumCap_1} using Gaussian test channels and $Q=\emptyset$, not the Berger-Tung inner bound. The reader may wonder: i) whether Gaussian test channels also exhaust the Berger-Tung inner bound for the vector Gaussian CEO problem that we study here, and ii) whether time-sharing is needed with the Berger-Tung scheme. This is addressed in Section~\ref{secIII_subsecB}, where it will be shown that the answer to both questions is positive. \qedblack
\end{remark}

\vspace{-1em}
\begin{remark}
For the converse proof of Theorem~\ref{th:GaussSumCap}, we derive an outer bound on the region described by~\eqref{eq:GaussSumCap_1}. In doing so, we use the de Bruijn identity, a connection between differential entropy and Fisher information, along with the properties of MMSE and Fisher information. By opposition to the case of quadratic distortion  for which the application of this technique was shown in~\cite{EU14} to result in an outer bound that is generally non-tight, Theorem~\ref{th:GaussSumCap} shows that the approach is successful in the case of logarithmic loss distortion measure, yielding a complete characterization of the region. Theorem~\ref{th:GaussSumCap} is also connected to recent developments on characterizing the capacity of multiple-input multiple-output (MIMO) relay channels in which the relay nodes are connected to the receiver through error-free finite-capacity links (i.e., the so-called cloud radio access networks). The reader may refer to~\cite[Theorem 4]{ZXYC16} where important progress is done, and \cite{EZSC17,EZSC17-2} where compress-and-forward with joint decompression-decoding is shown to be optimal under the constraint of oblivious relay processing. \qedblack
\end{remark}

\vspace{-1.9em}
\subsection{Gaussian Test Channels with Time-Sharing Exhaust the Berger-Tung Region}\label{secIII_subsecB}

In this section, we show that for the vector Gaussian CEO problem under logarithmic loss, the Berger-Tung coding scheme with Gaussian test channels and time-sharing achieves distortion levels that are not larger than any other coding scheme. That is, Gaussian test channels \textit{with} time-sharing exhaust the Berger-Tung region for this model.

\vspace{-0.5em}
\begin{definition}\label{defintion-continous-RD2-CEO}
For given tuple of auxiliary random variables $(V_1,\ldots,V_K,Q')$ with distribution $p(v_1,\ldots,v_K,q')$ such that $p(\dv x,\dv y_1,\ldots,\dv y_K,v_1,\ldots,v_K,q')$ factorizes as
\vspace{-0.5em}
\begin{equation}~\label{equation-joint-measure-auxiliary-random-variables-continous-2-CEO-problem}
p(q') p(\dv x) \prod\nolimits_{k=1}^K p(\dv y_k|\dv x) \prod\nolimits_{k=1}^K p(v_k|\dv y_k,q'),
\end{equation}
define $\mc {RD}^{\mathrm{II}}_{L}(V_1,\ldots,V_K,Q')$ as the set of all non-negative tuples $(R_1,\ldots,R_K,D)$ that satisfy, for all subsets $\mc S \subseteq \mc K$,
\begin{align*} 
\sum\nolimits_{k \in \mathcal{S}} R_k &\geq I(\dv Y_{\mc S};V_{\mc S}|V_{\mc S^{c}},Q') \\
D &\geq h(\dv X| V_1,\ldots,V_K,Q').
\end{align*}
{\fontsize{9.2}{11}\selectfont Also, let $\mc{RD}_\mathrm{L}^\mathrm{II} := \bigcup \mc{RD}_\mathrm{L}^\mathrm{II} (V_1,\ldots,V_K,Q')$ where the union is taken over all tuples $(V_1,\ldots,V_K,Q')$ with distributions that satisfy~\eqref{equation-joint-measure-auxiliary-random-variables-continous-2-CEO-problem}.}\qedblack
\end{definition}
\vspace{-1em}
\begin{proposition}\label{th:alternative}
$\mc {RD}_{L}^{*} = \bigcup\mc {RD}^\mathrm{II}_\mathrm{L}(V_1^{\mathrm{G}},\ldots,V_K^{\mathrm{G}},Q')$, 
where $\mc{RD}_\mathrm{L}^\mathrm{II} (\cdot)$ is as given in Definition~\ref{defintion-continous-RD2-CEO} and the superscript $\mathrm{G}$ is used to denote that the union is taken over Gaussian distributed $V_k^{\mathrm{G}}\sim p(v_k|\dv y_k,q')$ conditionally on $(\dv Y_k, Q')$.
\end{proposition}

\vspace{-1em}
\begin{proof} 
For the proof of Proposition~\ref{th:alternative}, it is sufficient to show that, for fixed Gaussian conditional distributions $\{p(u_k|\dv y_k)\}_{k=1}^K$,  the extreme points of the polytopes defined by~\eqref{eq:GaussSumCap_1} are \textit{dominated} by points that are in $\mc {RD}_{\mathrm{L}}^{\mathrm{II}}$ and which are achievable using Gaussian conditional distributions $\{p(v_k|\dv y_k, q')\}_{k=1}^K$. Hereafter, we give a brief outline of proof for the case $K=2$. The reasoning for $K \geq 2$ is similar and is provided in the extended version~\cite{PaperFull}. Consider the inequalities~\eqref{eq:GaussSumCap_1} with $Q=\emptyset$ and $(U_1,U_2) := (U^\mathrm{G}_1,U^\mathrm{G}_2)$ chosen to be Gaussian (see Theorem~\ref{th:GaussSumCap}). Consider now the extreme points of the polytopes defined by the obtained inequalities:
\vspace{-0.3em}
\begin{align*}
P_1 &= (0,0,I(\dv Y_1;U^\mathrm{G}_1|\dv X) + I(\dv Y_2;U^\mathrm{G}_2|\dv X)+ h(\dv X))\\
P_2 &= (I(\dv Y_1;U^\mathrm{G}_1),0,I(U^\mathrm{G}_2;\dv Y_2|\dv X) + h(\dv X|U^\mathrm{G}_1))\\
P_3 &= (0,I(\dv Y_2;U^\mathrm{G}_2),I(U^\mathrm{G}_1;\dv Y_1|\dv X) + h(\dv X|U^\mathrm{G}_2))\\
P_4 &= (I(\dv Y_1;U^\mathrm{G}_1),I(\dv Y_2;U^\mathrm{G}_2|U^\mathrm{G}_1), h(\dv X|U^\mathrm{G}_1,U^\mathrm{G}_2))\\
P_5 &= (I(\dv Y_1;U^\mathrm{G}_1|U^\mathrm{G}_2), I(\dv Y_2;U^\mathrm{G}_2), h(\dv X|U^\mathrm{G}_1,U^\mathrm{G}_2)),
\end{align*} \\[-1.3em]
where the point $P_j$ is a a triple $(R_1^{(j)}, R_2^{(j)}, D^{(j)})$. It is easy to see that each of these points is \textit{dominated} by a point in $\mc{RD}_{\mathrm{L}}^{\mathrm{II}}$, i.e., there exists $(R_1,R_2,D) \in \mc{RD}_{\mathrm{L}}^{\mathrm{II}}$ for which $R_1 \leq R_1^{(j)}$, $R_2\leq R_2^{(j)}$ and $D \leq D^{(j)}$. To see this, first note that $P_4$ and $P_5$ are both in $\mc{RD}_{\mathrm{L}}^{\mathrm{II}}$. Next, observe that the point $(0,0,h(\dv X))$ is in $\mc{RD}_{\mathrm{L}}^{\mathrm{II}}$, which is clearly achievable by letting $(V_1,V_2,Q') = (\emptyset,\emptyset,\emptyset)$, dominates $P_1$. Also, by using letting $(V_1,V_2,Q') = (U^\mathrm{G}_1,\emptyset,\emptyset)$,  we have that the point $(I(\dv Y_1;U_1),0, h(\dv X|U_1))$ is in $\mc{RD}_{\mathrm{L}}^{\mathrm{II}}$, and dominates the point $P_2$. A similar argument shows that $P_3$ is dominated by a point in $\mc{RD}_{\mathrm{L}}^{\mathrm{II}}$. The proof is terminated by observing that, for all above corner points, $V_k$ is set either equal $U^\mathrm{G}_k$ (which is Gaussian distributed conditionally on $\dv Y_k$) or a constant. 
\end{proof}

\vspace{-1.1em}
\begin{remark}
By opposition to the region $\mc {RD}^{\mathrm{I}}_{\mathrm{L}}$ described by the inequalities~\eqref{eq:GaussSumCap_1} for which we have shown that the time-sharing variable can be optimally set to $Q=\emptyset$, time-sharing may still be needed to exhaust the entire region $\mc {RD}^{\mathrm{II}}_{\mathrm{L}}$. On this aspect, we note that, from the proof of Proposition~\ref{th:alternative}, it is only implied that the corner points of this region are achieved with Gaussian test channels without time-sharing. To get the entire region, one needs to time-share Gaussian test channels. \qedblack
\end{remark}

\vspace{-1.7em}
\section{Quadratic Vector Gaussian CEO Problem with Determinant Constraint}\label{secIV}

{\fontsize{9.3}{11}\selectfont We turn to the case in which the distortion is measured under quadratic loss. In this case, the mean square error matrix is given by}
\vspace{-0.15em}
\begin{equation}\label{equation-mean-square-error}
\dv D^{(n)} := \frac{1}{n} \sum\nolimits_{i=1}^{n} \mathbb{E} \left[ (\dv X_i -  \hat{\dv X}_i) (\dv X_i -  \hat{\dv X}_i)^\dagger \right].
\end{equation}
Under a (general) error constraint of the form 
\vspace{-0.3em}
\begin{equation}
\dv D^{(n)} \preceq  \dv D,
\label{vector-Gaussian-ceo-quadratic-measure-matrix-constraint}
\end{equation}\\[-1.3em]
where $\dv D$ designates here a prescribed positive definite error matrix, a complete solution is still to be found in general. In what follows, we replace the constraint~\eqref{vector-Gaussian-ceo-quadratic-measure-matrix-constraint} with one on the \textit{determinant} of the error matrix $\dv D^{(n)}$, i.e.,
\vspace{-0.4em}
\begin{equation}~\label{vector-Gaussian-ceo-quadratic-measure-det-constraint}
|\dv D^{(n)}| \leq D,
\end{equation}\\[-1.3em]
($D$ is a scalar here). We note that since the error matrix $\dv D^{(n)}$ is minimized by choosing the decoding as  
\vspace{-0.1em}
\begin{equation*}~\label{equation-decoder}
\hat{\dv X}_i = \mathbb{E} [\dv X_i| \breve{\phi}^{(n)}_1  (\dv Y_1^n), \ldots, \breve{\phi}^{(n)}_K (\dv Y_K^n)],
\end{equation*}
where $\{\breve{\phi}^{(n)}_k\}_{k=1}^K$ denote the encoding functions, without loss of generality we can write~\eqref{equation-mean-square-error} as 
\vspace{-0.15em}
\begin{equation*}
\dv D^{(n)} = \frac{1}{n} \sum\nolimits_{i=1}^{n} \mathrm{mmse} (\dv X_i| \breve{\phi}^{(n)}_1  (\dv Y_1^n), \ldots, \breve{\phi}^{(n)}_K (\dv Y_K^n)). 
\end{equation*}

\begin{definition}
A rate-distortion tuple $(R_1,\ldots,R_K,D)$ is achievable for the quadratic vector Gaussian CEO problem with determinant constraint if there exist a blocklength $n$, $K$ encoding functions $\{\breve{\phi}^{(n)}_k\}^K_{k=1}$ such that
\vspace{-0.3em}
\begin{align*}
R_k &\geq \frac{1}{n}\log M^{(n)}_k, \quad \text{for } k=1,\ldots,K, \\
D &\geq 
\left|\frac{1}{n} \sum\nolimits_{i=1}^{n} \mathrm{mmse} (\dv X_i| \breve{\phi}^{(n)}_1  (\dv Y_1^n), \ldots, \breve{\phi}^{(n)}_K (\dv Y_K^n))  \right|. 
\label{equation-det-distortion}
\end{align*} 
\noindent The rate-distortion region $\mc{RD}^{\star}_\mathrm{Q}$ is defined as the union of all non-negative tuples $(R_1,\ldots,R_K,D)$ that are achievable. \qedblack
\end{definition}

\vspace{-0.3em}
\noindent The following lemma essentially states that Theorem~\ref{th:GaussSumCap} provides an outer bound on $\mc{RD}^\star_\mathrm{Q}$.

\vspace{-0.45em}
\begin{lemma}~\label{lemma-connection}
If $(R_1,\ldots,R_K,D) \in \mc{RD}^\star_\mathrm{Q}$, then  $(R_1,\ldots,R_K, \log ({{\pi}e})^{n_x} D) \in \mc{RD}^\mathrm{I}_\mathrm{L}$.  
\end{lemma}

\vspace{-1em}
\begin{proof}
The proof of Lemma~\ref{lemma-connection} is given in Section~\ref{secV_subsecB}.
\end{proof}

\vspace{-0.3em}
\noindent We are now ready to state the main result of this section, which is a complete characterization of the region $\mc{RD}^{\star}_\mathrm{Q}$.

\vspace{-0.4em}
\begin{theorem}~\label{theorem-rate-distortion-region-ceo-det-quadratic-constraint}
The rate-distortion region $\mc{RD}^\star_\mathrm{Q}$ of the quadratic vector Gaussian CEO problem with determinant constraint is given by the set of  all non-negative rate-distortion tuples $(R_1,\ldots,R_K,D)$ that satisfy, for all subsets $\mc S \subseteq \mc K$,
\vspace{-0.3em}
\begin{align*}
\log \frac{1}{D} \leq \!
\sum_{k\in \mathcal{S}} \! R_k+\log|\dv I-\dv\Omega_k \dv\Sigma_k|  
 + \log  \big|\dv\Sigma_{\dv x}^{-1} +  \sum_{k\in\mathcal{S}^{c}} \dv H_k^\dagger \dv\Omega_k \dv H_{k}\big|
\end{align*}\\[-0.5em]
for some $\dv 0 \preceq \dv\Omega_k \preceq \dv\Sigma_k^{-1}$, $k=1,\ldots,K$. 
\end{theorem}

\begin{proof} 
The proof of Theorem~\ref{theorem-rate-distortion-region-ceo-det-quadratic-constraint} is given in Section~\ref{secV_subsecC}.	
\end{proof}

\vspace{-1em}
\begin{remark}
It is believed that the approach of this section, which connects the quadratic vector Gaussian CEO problem to that under logarithmic loss, can also be exploited to possibly infer other new results on the  quadratic vector Gaussian CEO problem. Alternatively, it can also be used to derive new converses on the quadratic vector Gaussian CEO problem. For example, in the case of scalar sources, Theorem~\ref{theorem-rate-distortion-region-ceo-det-quadratic-constraint}, and Lemma~\ref{lemma-connection}, readily provide an alternate converse proof to those of~\cite{O05, PTR04} for this model. \qedblack 
\end{remark}

\vspace{-0.8em}
\section{Proofs}\label{secV}

\vspace{-0.2em}
\subsection{Proof of Proposition~\ref{proposition-continous-RD1-CEO}}\label{proof-continous-RD1-CEO}

First let us define the rate-information region $\mc{RI}_\mathrm{L}^\star$ for discrete memoryless sources as the closure of all rate-information tuples $(R_1,\ldots,R_K,\Delta)$ for which there exist a blocklength $n$, encoding functions $\{\phi^{(n)}_k\}^K_{k=1}$ and a decoding function $\psi^{(n)}$ such that  
\begin{align*}
R_k &\geq \frac{1}{n}\log M^{(n)}_k, \quad \text{for } k=1,\ldots,K, \\
\Delta &\leq \frac{1}{n} I(\dv X^n;\psi^{(n)}(\phi_1^{(n)}(\dv Y_1^n),\ldots,\phi_K^{(n)}(\dv Y_K^n))). 
\end{align*} 
It is easy to see that a characterization of $\mc{RI}_\mathrm{L}^\star$ can be obtained by using~\cite[Theorem 10]{CW14} and substituting distortion levels $D$ therein with $(H(\dv X) - D)$. More specifically, the region $\mc{RI}_\mathrm{L}^\star$ is given as in the following proposition. 

\vspace{-0.5em}
\begin{proposition}~\label{proposition-DM-RI-CEO}
The rate-information region $\mc{RI}_\mathrm{L}^\star$ of the vector DM CEO problem under logarithmic loss is given by the set of all non-negative tuples $(R_1,\ldots, R_K,D)$ that satisfy, for all subsets $\mc S \subseteq \mc K$,
\begin{align*}
\sum\nolimits_{k \in \mc S} R_k \geq \sum\nolimits_{k \in \mc S} I(\dv Y_k;U_k|\dv X,Q) - I(\dv X;U_{\mc S^c},Q) + \Delta,
\end{align*}
for some joint measure of the form $p(q) p(\dv x) \prod\nolimits_{k=1}^K p(\dv y_k|\dv x) \prod\nolimits_{k=1}^K p(u_k|\dv y_k,q)$. \qedblack
\end{proposition}

The region $\mc{RI}_\mathrm{L}^\star$ involves mutual information terms only (not entropies); and, so, using a standard discretization argument, it can be easily shown that a characterization of this region in the case of continuous alphabets is also given by Proposition~\ref{proposition-DM-RI-CEO}.

Let us now return to the vector Gaussian CEO problem under logarithmic loss that we study in this paper. First, we state the following lemma, whose proof is easy and is omitted for brevity.

\vspace{-0.5em}
\begin{lemma}~\label{lemma-relation-RD-RI}
$(R_1,\ldots,R_K,D) \in \mc{RD}^{\star}_\mathrm{L}$ if and only if $(R_1,\ldots,R_K,h(\dv X)-D) \in \mc{RI}^{\star}_\mathrm{L}$. \qedblack
\end{lemma}

\vspace{-0.5em}

\noindent For vector Gaussian sources, the region $\mc{RD}^{\star}_\mathrm{L}$ can be characterized using Proposition~\ref{proposition-DM-RI-CEO} and Lemma~\ref{lemma-relation-RD-RI}. This completes the proof.

\subsection{Proof of Converse of Theorem~\ref{th:GaussSumCap}}\label{secV_subsecA}

The proof of Theorem~\ref{th:GaussSumCap} relies on deriving an outer bound on the region $\mc{RD}_\mathrm{L}^\mathrm{I}$ given by Proposition~\ref{proposition-continous-RD1-CEO}. In doing so, we use the technique of~\cite[Theorem 8]{EU14} which relies on the de Bruijn identity and the properties of Fisher information and MMSE. 

\vspace{-0.3em}
\begin{lemma}{\cite{DCT91,EU14}}\label{lem:FI_Ineq}
Let $(\mathbf{X,Y})$  be a pair of random vectors with pmf $p(\mathbf{x},\mathbf{y})$. We have
\begin{align}
\log|(\pi e) \mathbf{J}^{-1}(\mathbf{X}|\mathbf{Y})|\leq h(\mathbf{X}|\mathbf{Y})\leq\log|(\pi e) \mathrm{mmse}(\mathbf{X}|\mathbf{Y})|,\nonumber
\end{align}
where the conditional Fisher information matrix is defined as
\begin{equation*}
\mathbf{J}(\mathbf{X}|\mathbf{Y}) := \mathrm{E}[\nabla \log p(\mathbf{X}|\mathbf{Y})\nabla\log p(\mathbf{X}|\mathbf{Y})^\dagger],
\end{equation*}
and the minimum mean squared error (MMSE) matrix is 
\begin{equation*}
\hspace{2em}\mathrm{mmse}(\mathbf{X}|\mathbf{Y}) := \mathrm{E}[(\dv X-\mathrm{E}[\dv X|\dv Y])(\dv X-\mathrm{E}[\dv X|\dv Y])^\dagger]. \hspace{2em}\bqed
\end{equation*}
\end{lemma}

\noindent First, we derive an outer bound on~\eqref{eq:GaussSumCap_1} as follows. For each $ q\in \mc{Q}$ and fixed pmf $\prod_{k=1}^Kp(u_k|\mathbf{y}_k,q)$, choose $\dv\Omega_{k,q}$, $k\in \mc K$, satisfying $\mathbf{0}\preceq \dv\Omega_{k,q} \preceq\mathbf{\Sigma}_{k}^{-1}$ such that 
\begin{equation}
\mathrm{mmse}(\mathbf{Y}_k|\mathbf{X}, U_{k,q},q) = \mathbf{\Sigma}_{k}-\mathbf{\Sigma}_{k} \dv\Omega_{k,q} \mathbf{\Sigma}_{k}.\label{eq:covB}
\end{equation}
Such $\dv\Omega_{k,q}$ always exists since, for all $q\in \mc Q$, $k\in \mc K$, we have
\begin{equation*}
\mathbf{0}\preceq\mathrm{mmse}(\mathbf{Y}_k|\mathbf{X},U_{k,q},q)\preceq \mathbf{\Sigma}_{\mb y_k|\dv x}=\mathbf{\Sigma}_{k}.\label{eq:CovKconst}
\end{equation*}

\noindent Then, for $k\in \mc K$ and $q\in \mc Q$, we have
\begin{align}
I(\mathbf{Y}_k&;U_k|\mathbf{X},Q=q)
= \log|(\pi e)\boldsymbol\Sigma_{k}| -h(\mathbf{Y}_k|\mathbf{X},U_{k,q},Q=q) \nonumber\\
&\stackrel{(a)}{\geq} \log|\boldsymbol\Sigma_{k}| -\log|\mathrm{mmse}(\mathbf{Y}_k|\mathbf{X},U_{k,q},Q=q)| \nonumber\\
&\stackrel{(b)}{=} -\log|\dv I- \dv\Omega_{k,q}\dv\Sigma_k|, \label{eq:firstIneq}
\end{align}
where $(a)$ is due to Lemma~\ref{lem:FI_Ineq}; and $(b)$ is due to \eqref{eq:covB}.

\noindent On the other hand, for $q\in \mc Q$ and $\mc {S}\subseteq \mc K$, we have
\begin{align}
h(\mathbf{X}|&U_{S^c,q},Q=q)
\stackrel{(a)}{\geq}  \log|(\pi e)\mathbf{J}^{-1}(\mathbf{X}|U_{S^c,q},q)| \nonumber\\
& \stackrel{(b)}{=}  \log 
\big| (\pi e)\big(\mathbf{\Sigma}_{\mb x}^{-1} + \sum\nolimits_{k\in\mathcal{S}^{c}}\mathbf{H}_{k}^{\dagger}
\dv\Omega_{k,q} \mathbf{H}_{k}\big)^{-1}\big|, \label{eq:FI_Ineq}  
\end{align}
where $(a)$ follows from Lemma~\ref{lem:FI_Ineq}; and for $(b)$, we use the connection of the MMSE and the Fisher information to show the following equality, whose proof is provided in the extended version~\cite{PaperFull}.
\begin{align}
\mathbf{J}(\mathbf{X}|U_{S^c,q},q) = \mathbf{\Sigma}_{\mb x}^{-1} + \sum\nolimits_{k\in\mathcal{S}^{c}}\mathbf{H}_{k}^\dagger
\dv\Omega_{k,q} \dv H_k \label{eq:Fischerequality}.
\end{align}

Next, we average~\eqref{eq:firstIneq} and~\eqref{eq:FI_Ineq} over the time sharing $Q$ and letting $\dv\Omega_k := \sum_{q\in \mathcal{Q}}p(q) \dv\Omega_{k,q}$, we obtain the lower bound 
\begin{align}
I(\mathbf{Y}_k;U_k|\mathbf{X},Q) &= \sum\nolimits_{q\in \mathcal{Q}}p(q)I(\mathbf{Y}_k;U_k|\mathbf{X},Q=q)\nonumber\\
&\stackrel{(a)}{\geq} - \sum\nolimits_{q\in \mathcal{Q}}p(q)\log|\dv I- \dv\Omega_{k,q} \dv\Sigma_k| \nonumber\\
&\stackrel{(b)}{\geq} -\log |\dv I- \dv\Omega_k \dv\Sigma_k|,\label{eq:logDetProp2}
\end{align}
where $(a)$ follows from~\eqref{eq:firstIneq}; and $(b)$ follows from the concavity of the log-det function and Jensen's Inequality. 

\noindent Besides,  we can derive the following lower bound
{\footnotesize
\begin{align}
h(\mathbf{X}|U_{S^c},Q)
&\stackrel{(a)}{\geq} \sum_{q\in \mathcal{Q}}p(q) 
\log 
\big|(\pi e)\big(\mathbf{\Sigma}_{\mb x}^{-1} + \sum\nolimits_{k\in\mathcal{S}^{c}}\mathbf{H}_{k}^\dagger
\dv\Omega_{k,q} \mathbf{H}_{k}\big)^{-1}\big| \nonumber\\
&\stackrel{(b)}{\geq} \log \big|(\pi e)\big(\mathbf{\Sigma}_{\mb x}^{-1} + \sum\nolimits_{k\in\mathcal{S}^{c}}\mathbf{H}_{k}^\dagger  \dv\Omega_k \mathbf{H}_{k}
\big)^{-1}\big|, \label{eq:secondtIneq_4}
\end{align}}
\hspace{-2mm} where $(a)$ is due to~\eqref{eq:FI_Ineq}; and $(b)$ is due to the concavity of the log-det function and Jensen's inequality. 

Finally, the outer bound on $\mc{RD}_{\mathrm{L}}^{\star}$ is obtained by applying~\eqref{eq:logDetProp2} and~\eqref{eq:secondtIneq_4} in~\eqref{eq:GaussSumCap_1}, noting that $\dv\Omega_k = \sum_{q\in \mathcal{Q}}p(q) \dv\Omega_{k,q} \preceq\mathbf{\Sigma}_{k}^{-1}$ since $\mathbf{0} \preceq \dv\Omega_{k,q} \preceq\mathbf{\Sigma}_{k}^{-1}$, and taking the union over $\dv\Omega_k$ satisfying \mbox{$\mathbf{0} \preceq \dv\Omega_k \preceq\mathbf{\Sigma}_{k}^{-1}$}. 

\subsection{Proof of Lemma~\ref{lemma-connection}}\label{secV_subsecB}

Let a tuple $(R_1,\ldots,R_K,D) \in \mc{RD}^{\star}_\mathrm{Q}$ be given. Then, there exist a blocklength $n$, $K$ encoding functions $\{\breve{\phi}^{(n)}_k\}^K_{k=1}$ and a decoding function $\breve{\psi}^{(n)}$ such that    
\begin{align}
R_k &\geq \frac{1}{n}\log M^{(n)}_k, \quad \text{for } k=1,\ldots,K, \nonumber\\
D &\geq 
\big|\frac{1}{n} \sum\nolimits_{i=1}^{n} \mathrm{mmse} (\dv X_i| \breve{\phi}^{(n)}_1  (\dv Y_1^n), \ldots, \breve{\phi}^{(n)}_K (\dv Y_K^n))  \big|. 
\label{equation-log-quadrtic-1}
\end{align} 

\noindent We need to show that there exist $(U_1,\ldots,U_K,Q)$ such that, for all subsets $\mc S \subseteq \mc K$,
\begin{equation}
\sum_{k \in \mc S} R_k + \log ({{\pi}e})^{n_x} D \geq \sum_{k \in \mc S} I(U_k;\dv Y_k|\dv X,Q) + h(\dv X| U_{\mc S^c},Q). 
\label{required-existence-condition-proof-lemma}
\end{equation}

\noindent Let us define  
\begin{equation*}
\bar{\Delta}^{(n)} := \frac{1}{n} h(\dv X^n | \breve{\phi}^{(n)}_1  (\dv Y_1^n), \ldots, \breve{\phi}^{(n)}_K (\dv Y_K^n)).
\end{equation*}
\noindent  It is easy to justify that expected distortion $\bar{\Delta}^{(n)}$ is achievable under logarithmic loss (see Proposition~\ref{proposition-continous-RD1-CEO}). Then, following straightforwardly the lines in the proof of~\cite[Theorem 10]{CW14}, we have
\begin{align}~\label{equation-log-quadrtic-2}
\sum\nolimits_{k \in \mc S} R_k  \geq & \: \sum\nolimits_{k \in \mc S} \frac{1}{n} \sum\nolimits_{i=1}^{n} I(\dv Y_{k,i};U_{k,i}|\dv X_i,Q_i) \nonumber\\
& + \frac{1}{n} \sum\nolimits_{i=1}^{n} h(\dv X_i | U_{\mc S^c,i},Q_i) - \bar{\Delta}^{(n)}.
\end{align}

Next, we upper bound $\bar{\Delta}^{(n)}$ in terms of $D$. Letting $J_{\mc K}:=(\breve{\phi}^{(n)}_1  (\dv Y_1^n), \ldots, \breve{\phi}^{(n)}_K (\dv Y_K^n))$, we have
\begin{align} 
\hspace{-0.3em}\bar{\Delta}^{(n)} &=\: \frac{1}{n} h(\dv X^n | J_{\mc K} ) =  \frac{1}{n} \! \sum\nolimits_{i=1}^{n} \! h(\dv X_i|\dv X_{i+1}^n, J_{\mc K} ) \nonumber\\
&= \:  \frac{1}{n} \sum\nolimits_{i=1}^{n} h(\dv X_i - \mathbb{E}[\dv X_i | J_{\mc K} ] \big| \dv X_{i+1}^n, J_{\mc K}) \nonumber\\
& \stackrel{(a)}{\leq} \:  \frac{1}{n} \sum\nolimits_{i=1}^{n} h(\dv X_i - \mathbb{E}[\dv X_i|J_{\mc K}] ) \nonumber\\
& \stackrel{(b)}{\leq} \: \frac{1}{n} \sum\nolimits_{i=1}^{n} \log (\pi e)^{n_x} \left| \mathrm{mmse}(\dv X_i|J_{\mc K}) \right| \nonumber\\
& \stackrel{(c)}{\leq}  \: \log ({{\pi}e})^{n_x} \big| \frac{1}{n} \sum_{i=1}^{n}  \mathrm{mmse}(\dv X_i|J_{\mc K}) \big| \stackrel{(d)}{\leq}    \log ({{\pi}e})^{n_x} D,
 \label{equation-log-quadrtic-3} 
\end{align}
where  $(a)$  holds since conditioning reduces entropy; $(b)$ is due to the maximal differential entropy lemma; $(c)$ is due to the convexity of the log-det function and Jensen's inequality; and $(d)$ is due to~\eqref{equation-log-quadrtic-1}. 

\noindent Combining~\eqref{equation-log-quadrtic-3} with~\eqref{equation-log-quadrtic-2}, and using standard arguments for single-letterization, we get~\eqref{required-existence-condition-proof-lemma}; and this completes the proof of the lemma.  

\vspace{-0.5em}
\subsection{Proof of Theorem~\ref{theorem-rate-distortion-region-ceo-det-quadratic-constraint}}\label{secV_subsecC}

The proof is as follows. By Lemma~\ref{lemma-connection} and Proposition~\ref{th:alternative}, there must exist Gaussian test channels $(V^\mathrm{G}_1,\ldots,V^\mathrm{G}_K)$ and a time-sharing random variable $Q'$, with joint distribution that factorizes as $p(q') p(\dv x) \prod\nolimits_{k=1}^K p(\dv y_k|\dv x) \prod\nolimits_{k=1}^K p(v_k|\dv y_k,q')$, such that the following holds for all subsets $\mc S \subseteq \mc K$, 
\begin{align} 
\sum\nolimits_{k \in \mathcal{S}} R_k &\geq I(\dv Y_{\mc S};V^\mathrm{G}_{\mc S}|V^\mathrm{G}_{\mc S^{c}},Q')\\
\log ((\pi e)^{n_x} D &\geq h(\dv X| V^\mathrm{G}_1,\ldots,V^\mathrm{G}_K,Q').
\label{optimal-rate-distortion region}
\end{align}

\noindent This is clearly achievable by the Berger-Tung coding scheme with Gaussian test channels and time-sharing $Q'$, since the achievable error matrix under quadratic distortion has determinant that satisfies
\begin{align}
\log((\pi e)^{n_x} | \mathrm{mmse}(\dv X|V^\mathrm{G}_1,\ldots,V^\mathrm{G}_K,Q')|) = h(\dv X|V^G_1,\ldots,V^G_K,Q').\nonumber
\end{align}

\noindent The above shows that the rate-distortion region of the quadratic vector Gaussian CEO problem under determinant constraint is given by~\eqref{optimal-rate-distortion region}, i.e., $\mc{RD}^{\mathrm{II}}_\mathrm{L}$ (with distortion parameter $\log({\pi}e)^{n_x}D$). Recalling that $\mc{RD}^{\mathrm{II}}_L=\mc{RD}^{\mathrm{I}}_\mathrm{L}=\mc{RD}^{\star}_\mathrm{L}$, and substituting in Theorem~\ref{th:GaussSumCap} using distortion level $\log({\pi}e)^{n_x}D$ completes the proof.

\bibliographystyle{IEEEtran}
\bibliography{IEEEabrv,mybibfile}
\end{document}